\documentclass[conference,a4paper]{IEEEtran}

\usepackage[
  bottom=4.365cm, top=2cm, right =1.45cm, left = 1.45cm
]{geometry}

\usepackage[utf8]{inputenc} 
\usepackage[T1]{fontenc}
\usepackage{url}
\usepackage{ifthen}
\usepackage{cite}
\usepackage[cmex10]{amsmath} 
                             
\usepackage{amssymb}
\usepackage[nolist]{acronym} 
\usepackage{nicefrac}
\usepackage{xcolor}
\usepackage[normalem]{ulem}

\usepackage{booktabs}
\usepackage{nicematrix}
\usepackage{mathtools}

\usepackage{graphicx}

\usepackage[hidelinks]{hyperref}
\usepackage[capitalise]{cleveref}
\crefname{construction}{Construction}{Constructions}
\crefname{proposition}{Proposition}{Propositions}
\crefname{algorithm}{Algorithm}{Algorithm}
\crefname{property}{Property}{Property}
\crefname{calc}{Calculation Rule}{Calculation Rules}
\crefname{subsection}{Section}{Sections}
\Crefname{subsection}{Section}{Sections}

\usepackage{amsthm}

\newtheorem{proposition}{Proposition}

\newtheorem{lemma}{Lemma}

\theoremstyle{definition}
\newtheorem{definition}{Definition}

\newtheorem{example}{Example}
\newtheorem{calc}{Calculation Rule}
\newtheorem{property}{Property}

\theoremstyle{remark}
\newtheorem*{remark}{Remark}

\usepackage{tikz}
\usetikzlibrary{arrows,calc,shapes,graphs,graphs.standard,quotes,arrows.meta,decorations.markings,positioning,fit,positioning,hobby,decorations.text,decorations.pathreplacing, patterns}
\usetikzlibrary{spy}

\newenvironment{writereplace}{
    \begingroup
    \let\writeold=\write
    \def\write{\immediate\writeold}
}{
    \endgroup
}%

\usepackage{pgfplots}
\pgfplotsset{compat=1.18}
\usepgfplotslibrary{groupplots}

\usepackage[
    linesnumbered,
    titlenumbered,
    ruled,
    vlined,
    resetcount
]{algorithm2e}

\makeatletter
    \patchcmd\algocf@Vline{\vrule}{\vrule \kern-0.4pt}{}{}
    \patchcmd\algocf@Vsline{\vrule}{\vrule \kern-0.4pt}{}{}
    \patchcmd{\@algocf@start}{-1.5em}{-1.25em}{}{}
\makeatother

\SetAlFnt{\small}
\SetAlCapNameFnt{\footnotesize}
\SetAlCapFnt{\footnotesize}

\AtBeginEnvironment{algorithm}{\SetAlgoLined}
\SetKwInOut{Input}{\small Input}
\SetKwInOut{Output}{\small Output}

\renewcommand{\mid}{\,|\,}

\newcommand{\ml}{_{\mathrm{ML}}}
\newcommand{\lw}{_{\mathrm{LW}}}

\newcommand{\vc}{\boldsymbol{c}}

\newcommand{\ve}{\boldsymbol{e}}

\newcommand{\vu}{\boldsymbol{u}}
\newcommand{\vv}{\boldsymbol{v}}

\newcommand{\vy}{\boldsymbol{y}}

\newcommand{\calE}{\mathcal{E}}

\newcommand{\calI}{\mathcal{I}}
\newcommand{\calJ}{\mathcal{J}}
\newcommand{\calK}{\mathcal{K}}

\newcommand{\calP}{\mathcal{P}}

\newcommand{\calR}{\mathcal{R}}
\newcommand{\calS}{\mathcal{S}}

\newcommand{\Px}{P}

\newcommand{\Cod}{\mathcal{C}}

\newcommand{\dmin}{d_{\mathrm{min}}}
\newcommand{\ddes}{d_{\mathrm{des}}}
\newcommand{\List}{\mathcal{L}}

\newcommand{\Ind}{{\mathcal{I}}}
\newcommand{\Kcal}{\mathcal{K}}
\newcommand{\one}{\boldsymbol{1}}
\newcommand{\zero}{\boldsymbol{0}}
\newcommand{\pws}{\calP}

\newcommand{\vrho}{\boldsymbol{\rho}}

\newcommand{\Pcy}{\Px}

\newcommand{\chd}{\hat{y}}

\newcommand{\vchat}{\boldsymbol{\hat{c}}}
\newcommand{\vcprime}{\boldsymbol{c^\prime}}

\newcommand{\vchd}{\boldsymbol{\hat{y}}}

\newcommand{\vcml}{\vc\ml} %

\newcommand{\Ball}{\mathcal{B}}

\newcommand{\Ftwo}{\mathbb{F}_{2}}

\newcommand{\Vcov}{\mathcal{V}}

\newcommand{\ttest}{\boldsymbol{\tau}}
\newcommand{\valpha}{\boldsymbol{\alpha}}

\newcommand{\Test}{\mathcal{T}}

\newcommand{\Nat}{\mathbb{N}}

\newcommand{\fcor}{f_{\mathrm{c}}}
\newcommand{\ferr}{f_{\mathrm{e}}}
\newcommand{\Fcor}{F_{\mathrm{c}}}
\newcommand{\Ferr}{F_{\mathrm{e}}}

\newcommand{\betaos}[1]{\beta_{b,(#1)}}
\newcommand{\betabos}[2]{\beta_{#1,(#2)}}
\newcommand{\gammaos}[1]{\gamma_{n-b,(#1)}}
\newcommand{\gammabos}[2]{\gamma_{n-#1,(#2)}}
\newcommand{\gammanbos}[2]{\gamma_{#1,(#2)}}
\newcommand{\vbeta}{\boldsymbol{\beta}}
\newcommand{\vgamma}{\boldsymbol{\gamma}}

\newcommand{\vell}{\boldsymbol{\ell}}

\newcommand{\Pat}{\calR}
\newcommand{\pat}{\vrho}
\newcommand{\patprime}{\boldsymbol{\rho^\prime}}
\newcommand{\Patprime}{\Pat^\prime}
\newcommand{\wmax}{w_{\mathrm{H}}^{\mathrm{max}}}
\newcommand{\wlw}{w_\mathrm{L}}
\newcommand{\wham}{w_{\mathrm{H}}}
\newcommand{\dham}{d_\mathrm{H}}
\newcommand{\Pna}{P}
\newcommand{\Pbdd}{P_{\mathrm{BDD}}}

\newcommand{\Patchase}{\Pat\chase}
\newcommand{\chase}{_{\mathrm{Chase}}}
\newcommand{\reschase}{_{\mathrm{restr}}}
\newcommand{\Patlw}{\Pat\lw}
\newcommand{\Patlwprec}{\Patlw^{\prec\pat}}
\newcommand{\npat}{q}
\newcommand{\imax}{i^{\mathrm{max}}}

\newcommand{\Pcl}{P_{\Pat}} %
\newcommand{\Pcov}{P_\mathrm{cov}}
\newcommand{\Padd}{P^{\mathrm{add}}}

\newcommand{\MLP}{\operatorname{OCP}}

\newcommand{\Scond}{\calS}
\newcommand{\df}{\delta_{\Scond}}
\newcommand{\Icond}{\calI_{\Scond}}
\newcommand{\Jcond}{\calJ_{\Scond}}
\newcommand{\Rcond}[1]{\mathbb{R}^{#1}_{\geq 0,\,\Scond}}

\newcommand{\qsigma}{q_{\sigma}}
\newcommand{\Qsigma}{Q_{\sigma}}

\newcommand{\tikzgets}{\tikz[baseline=-0.5ex] \draw[-{Straight Barb[scale=0.8]}, yshift=0.05ex] (0.3,0) -- (0,0);}
\renewcommand{\gets}{\mathrel{\tikzgets}}

\DeclareMathOperator*{\argmax}{arg\,max}
\DeclareMathOperator*{\argmin}{arg\,min}

\SetKwProg{Fn}{subroutine}{:}{end}
\let\oldnl\nl%
\newcommand{\nonl}{\renewcommand{\nl}{\let\nl\oldnl}}%
\newcommand{\algrule}[1][.2pt]{\par\vskip.5\baselineskip\hrule height #1\par\vskip.5\baselineskip}

\definecolor{mittelblau}{RGB}{0, 81, 158} %
\definecolor{anthrazit}{RGB}{62, 68, 76} %
\definecolor{hellblau}{RGB}{0, 190, 255} %
\definecolor{violettblau}{cmyk}{0.9, 0.6, 0, 0}
\definecolor{rot}{RGB}{238, 28 35}
\definecolor{apfelgruen}{RGB}{140, 198, 62}
\definecolor{gelb}{RGB}{255, 229, 0}
\definecolor{orange}{RGB}{244, 111, 33}
\definecolor{pink}{RGB}{237, 0, 140}
\definecolor{lila}{RGB}{128, 10, 145}
\definecolor{hellgrau}{RGB}{224, 224, 224}
\definecolor{mittelgrau}{RGB}{128, 128, 128}
\definecolor{dunkelgrau}{RGB}{80,80,80}
\definecolor{darkapfelgruen}{RGB}{100, 140, 45}

\def\colorChase{apfelgruen}
\def\colorRC{hellblau}
\def\colorLW{violettblau}
\def\colorMCOC{rot}
\def\colorGreedy{black}

\def\markChase{triangle}
\def\markRC{triangle}
\def\markLW{diamond}
\def\markMCOC{asterisk}
\def\markGreedy{square}

\def\markChase{triangle}
\def\markRC{triangle}
\def\markLW{Mercedes star}
\def\markMCOC{Mercedes star flipped}
\def\markGreedy{square}
\def\markML{diamond}

\def\lineStyleA{densely dotted}
\def\lineStyleB{densely dashed}
\def\lineStyleC{densely dashdotted}

\begin{acronym}
\acro{AWGN}{additive white Gaussian noise}
\acro{BDD}{bounded-distance decoder}
\acro{BCH}{Bose--Ray-Chaudhuri--Hocquenghem}
\acro{BI-AWGN}{binary-input additive white Gaussian noise}
\acro{BLER}{block error rate}
\acro{BPSK}{binary phase shift keying}
\acro{CDF}{cumulative distribution function}
\acro{GMD}{generalized minimum distance}
\acro{GRAND}{guessing random additive noise decoding}
\acro{HDD}{hard-decision decoding}
\acro{LER}{list error rate}
\acro{LLR}{log-likelihood ratio}
\acro{LW}{logistic weight}
\acro{MCOC}{maximally covering ordered candidate}
\acro{ML}{maximum likelihood}
\acro{RS}{Reed--Solomon}
\acro{TPD}{turbo product decoding}
\acro{SNR}{signal-to-noise-ratio}
\end{acronym}

\begin{document}

\title{Chase-like Decoding: Test-Pattern Design \\and Performance Analysis} 

\author{\IEEEauthorblockN{Tim Janz, Simon Obermüller, Andreas Zunker, and Stephan ten Brink}
\IEEEauthorblockA{Institute of Telecommunications, 
University of Stuttgart, Germany\\
\{janz, obermueller, zunker, tenbrink\}@inue.uni-stuttgart.de}
}

\maketitle

\begin{abstract}
    Chase-like decoding algorithms are a popular choice for soft-input decoding of algebraic codes. We evaluate different test-pattern sets for Chase-like decoding. Structured sets, such as Chase-II patterns or patterns chosen by logistic weight, are analyzed using order statistics, while arbitrary sets are evaluated by calculating covered-space probabilities and by performing Monte Carlo simulation. We further propose an algorithm that designs test-pattern sets to cover likely error patterns, achieving comparable performance with half the number of test patterns compared with conventional sets for high-rate BCH codes.
\end{abstract}

\section{Introduction}
Due to their good minimum distance and efficient decoding, binary \ac{BCH} codes \cite{hocquenghem1959,bose1960} are widely used in communication systems \cite{lin2004}. They are often employed in concatenated coding schemes \cite{forney1966a} to improve performance with limited additional latency and complexity. Examples include data-center interconnects and the next Ethernet standard, which uses a KP4 (544,514) \ac{RS} outer code \cite{matuz2025}. In long-haul optical communications, \ac{BCH} codes are also used as component codes of generalized product codes \cite{smith2012a,sukmadji2022,shehadeh2025}, such as oFEC \cite{ofec2023}.

For \ac{BCH} codes, \acl{HDD} with a \ac{BDD} exhibits a substantial performance gap to \ac{ML} decoding. Soft-input decoding can reduce this gap, either with hard output for concatenated schemes with hard-input outer codes or with soft output for iterative component decoding, as in \ac{TPD} \cite{pyndiah1998}.

In both settings, the Chase-II algorithm is widely used \cite{chase1972}.
 Given a channel output, it generates a list of candidate codewords $\List$ by applying a \ac{BDD} to $2^p$ testwords. 
In the context of \ac{GRAND}, patterns based on their logistic weight were introduced \cite{duffy2022}.
They can also be used as test patterns in Chase-like decoding and achieve good performance \cite{shen2025}.
Few works consider alternative test-pattern sets. 
In \cite{tokushige2003} and \cite{tokushige2010}, covering codes \cite{cohen1997} are used to design test patterns for medium-rate \ac{BCH} codes, while \cite{cardinale2026} recently proposed a heuristic design for high-rate \ac{BCH} codes similar to the one presented here.

The performance of Chase-II and \ac{GMD} decoding \cite{forney1966} has been analyzed in \cite{agrawal2000,fossorier2002,he2025}. However, a comprehensive comparison of different test-pattern sets, particularly for high-rate \ac{BCH} codes, is, to the best of our knowledge, still missing.

In this paper, we present three methods to determine the \ac{LER}, i.e., the probability that the sent codeword is not in the list generated by the Chase-like decoder, and show that they yield equivalent results.
The first method, based on order statistics, was introduced in \cite{fossorier2002} building on \cite{agrawal2000}. We extend it to Chase-II patterns with restricted Hamming weight and patterns bounded by logistic weight. The second method uses the covered space proposed in \cite{janz2025} to compute the \ac{LER} for arbitrary test-pattern sets, while the third is a Monte Carlo simulation as a reference.
Furthermore, we propose an algorithm for designing efficient test-pattern sets that slightly outperform logistic-weight patterns. Compared with commonly used Chase-II patterns of restricted Hamming weight, the designed sets require less than half as many test patterns to achieve the same performance.

\section{Preliminaries}

Let $[n]\triangleq\{1,2,\ldots,n\}$ and $[m,n]\triangleq\{m,m+1,\ldots,n\}$ for $m<n$ with $n \in \Nat$ and $m\in \Nat_0$. 
We denote the power set of an index set $\Kcal \subseteq [n]$ as $\pws_{\Kcal}$ with $|\pws_{\Kcal}|=2^{|\Kcal|}$.
The subvector of $\vv \in \Ftwo^n$ containing only the elements at positions $\Kcal \subseteq [n]$ of $\vv$ is given by $\vv_\Kcal\in\Ftwo^{|\Kcal|}$.
The vector of length $n$ with ones at positions $\Kcal\subseteq [n]$ and zeros at positions $[n]\setminus \Kcal$ is denoted by $\one_{\Kcal}$. For a vector $\vu\in\Ftwo^n$, the set of positions equal to one is $\calI_{\vu}$ and $\imax_{\vu}=\max(\calI_{\vu})$.
Let $\vu \oplus \vv$ be the element-wise addition over $\Ftwo$ for two vectors $\vu,\vv\in\Ftwo^n$. The Hamming weight of a vector $\vu\in\Ftwo^n$ and the Hamming distance of the vectors $\vu,\vv\in\Ftwo^n$ are denoted by $\wham(\vu)$ and $\dham(\vu,\vv)$, respectively.
Let ${\Ball_r(\vu) \triangleq \{\vv \in \Ftwo^n \mid d_\mathrm{H}(\vu, \vv) \leq r\}}$ be the Hamming ball of radius $r$ that is centered at $\vu$.
The logistic weight of a vector $\vu$ is defined as $\wlw(\vu) = \sum_{i=1}^{n} iu_i $.

\subsection{Channel Model}\label{subsec:channel}
We consider transmission over a \ac{BI-AWGN}  channel $Y=X+Z$, where ${Z\sim \mathcal{N}(0,\sigma^2)}$ and ${X\in\{+1,-1\}}$  is the \ac{BPSK} modulated channel input with ${0\mapsto+1}$ and $1\mapsto-1$. 
As channel codes, we use $(n,k,t)$ primitive, narrow-sense, binary \ac{BCH} codes of length $n$, dimension $k$ and error correction capability ${t=\left\lfloor\frac{\ddes-1}{2}\right\rfloor}$, where the design distance $\ddes$ is a lower bound on the minimum distance $\dmin$. 

For a given channel output $\vy\in \mathbb{R}^n$, we can compute the channel \ac{LLR} as ${\vell  = \frac{2}{\sigma^2} \, \vy}$ and define the reliability as $\valpha\triangleq|\vell|$. Throughout this paper, we assume that the received channel output is already sorted with respect to the corresponding reliabilities such that $\alpha_1\leq\alpha_2\leq\ldots\leq \alpha_n$. 
The hard decision is denoted by $\vchd$, where $\chd_i=1$ for $y_i<0$ and $\chd_i=0$ otherwise.

For a transmitted codeword $\vc\in\Ftwo^n$, the error vector is given by $\ve=\vc\oplus\vchd$.
Let $\calE_b$ denote the event that the hard decision is erroneous at $b$ positions.
For the \ac{BI-AWGN} channel, the probability of $b$ bit flips is
\begin{equation}\label{eq:prob_i_flips}
  P(\calE_b) = \binom{n}{b} Q(1/\sigma)^{b} [1-Q(1/\sigma)]^{(n-b)} ,
\end{equation}
with $Q(\cdot)$ being the tail distribution function
of $\mathcal{N}(0,1)$.

For a \ac{BDD} with error correction capability $t$, the probability that the transmitted codeword is decoded correctly can be calculated~by
\begin{equation}\label{eq:p_bdd}
  \Pbdd \triangleq\Pbdd(\vchat = \vc) = \sum_{i=0}^t P(\calE_i).
\end{equation}

\subsection{Chase-like Decoding}

A popular approach for incorporating soft information into hard-decision decoding is the Chase-II algorithm \cite{chase1972}.
More generally, a Chase-like algorithm with a predetermined test-pattern set $\Pat$ operates as follows.

First, the hard-decision vector $\vchd = (\chd_1, \ldots, \chd_n)$ is formed. 
The corresponding set of testwords is $\Test = \left\{ \vchd \oplus \pat \;\middle|\; \pat \in \Pat \right\}$, where $\Pat$ is the set of test patterns with cardinality $\npat \triangleq |\Pat|$. 
Each testword $\ttest \in \Test$ is then decoded by a \ac{BDD}, and the obtained unique codewords form the candidate list $\List$.
For hard-output decoding, the most likely codeword is selected as $\vchat = \arg\max_{\vcprime \in \List} \Pcy(\vcprime \mid \vy)$.
Alternatively, the remaining list candidates can be used to generate soft outputs \cite{pyndiah1998,janz2025}.

Chase-like algorithms differ in the choice of the test-pattern set $\Pat$.
The Chase-II algorithm considers all $2^p$ bit flips within the $p$ least reliable positions with 
\begin{equation*}
    \Patchase = \left\{\one_{\calK} \;\middle|\; \calK \in \pws_{[p]}\right\}.
\end{equation*}
A common alternative that restricts the patterns to at most $\wmax$ bits among the $p$ least reliable positions is the \emph{restricted Chase pattern set} given by 
\begin{equation*}
    \Pat\reschase = \left\{\one_{\calK} \;\middle|\; \calK \in \pws_{[p]},\, |\calK| \leq \wmax \right\}.
\end{equation*}

Recently, patterns ordered by increasing \ac{LW} were proposed for \ac{GRAND} \cite{duffy2022}.
We refer to them as \emph{logistic-weight patterns} $\Patlw$ and consider them in Subsection~\ref{subsec:LER_OS} in detail.

\subsection{Error Probability of Chase-like Decoding}

The error probability of a Chase-like decoding algorithm with test-pattern set $\Pat\subseteq\Ftwo^n$ can be decomposed as
\begin{align}\label{eq:bler_bound}
    \Pcl(\vchat \neq \vc) 
    &= \Pcl ( \vchat\neq \vc, \vc\in \List) + \Pcl ( \vchat\neq \vc, \vc\notin \List)  \nonumber\\
    &\vphantom{\overset{\raisebox{0pt}[\height][2pt]{\smash{\scalebox{0.7}{$\mathrm{(b)}$}}}}{\leq}}
    \overset{\raisebox{0pt}[\height][2pt]{\smash{\scalebox{0.7}{$\mathrm{(a)}$}}}}{=} \Pcl ( \vchat\neq \vc, \vc\in \List) + \Pcl ( \vc\notin \List)  \nonumber\\
    &\overset{\raisebox{0pt}[\height][2pt]{\smash{\scalebox{0.7}{$\mathrm{(b)}$}}}}{\leq} P(\vcml \neq \vc) + \Pcl ( \vc\notin \List),
\end{align}
where $\mathrm{(a)}$ follows since $\vchat\in\List$ always holds, and $\mathrm{(b)}$ follows because the term $\Pcl ( \vchat\neq \vc, \vc\in \List)$ is maximized for $\List=\Cod$, yielding the \ac{BLER} of \ac{ML} decoding.
We call $\Pcl ( \vc\notin \List)$ the \acf{LER}, which provides a tight approximation of the \ac{BLER} for $P (\vcml \neq \vc) \ll \Pcl (\vc \notin \List)$ \cite{agrawal2000, fossorier2002}.

\subsection{Covered Space Probability Calculation}
For the probability $\Pna(\vv \mid \vy)$ that a vector $\vv$ was transmitted given the received vector $\vy$, without considering codebook information, we obtain
\begin{equation*}
    \Pna(\vv\mid\vy) = 
    \prod\limits_{i:\,v_i=\chd_i}\gamma_i
    \cdot\prod\limits_{i:\,v_i\neq\chd_i} (1-\gamma_i),
\end{equation*}
with $\gamma_i = \Pna(v_i=\chd_i\mid y_i) = \frac{1}{1+e^{-\ell_i}}$.
This probability should not be confused with the codebook-aware posterior probability $P(\vc\mid \vy)$, which is only non-zero for $\vc\in \Cod$ and not considered in this paper.

The subset of the ambient space $\Ftwo^n$ that is explored for codewords by a Chase-like algorithm is called the \emph{covered space} $\Vcov$ as in  \cite{janz2025}. For Chase-like decoding, it is given by
\begin{equation*}    
    \Vcov = \bigcup_{\pat \in \Pat} \Ball_t (\pat).
\end{equation*}
In other words, a pattern $\pat$ covers an error $\ve$ if $\ve \in \Ball_t(\pat)$ for a \ac{BDD} with error correction capability $t$.
The probability of the covered space is defined as 
\begin{equation*}
\Pcov(\vy) = \sum_{\vv\in\Vcov} P(\vv\mid \vy)
= P( \vc\in \List \mid \vy),
\end{equation*}
which coincides with the probability that the transmitted codeword $\vc$ is in the list $\List$ found by the Chase-like algorithm, since $\List = \Vcov \cap \Cod$.
An efficient way for evaluating $\Pcov (\vy)$ for the Chase-II algorithm using Hamming balls is given in \cite{janz2025}.

\section{Performance Evaluation Methods}\label{sec:methods}

In the following, we consider three methods for evaluating the \ac{LER} of a given test-pattern set and discuss two of them in detail. The third method directly estimates the \ac{LER} and \ac{BLER} by Monte Carlo simulation and applies to arbitrary test-pattern sets $\Pat$.

To quantify the contribution of an additional set of test patterns $\Patprime$ to an existing set $\Pat$, we define

\begin{equation*}
    \Padd_{\Pat} (\Patprime) \triangleq P_{\Pat \cup \Patprime} (\vc \in \List) - P_{\Pat} (\vc \in \List).
\end{equation*}
With slight abuse of notation, the probability added by a single test pattern $\pat$ is denoted as $\Padd_{\Pat} (\pat) = \Padd_{\Pat} (\{\pat\})$.

\subsection{List Error Rate with Order Statistics}\label{subsec:LER_OS}
Rather than computing $\Pcl(\vc\notin\List)$ directly as in \cite{fossorier2002}, we consider its complement $\Pcl(\vc\in\List)$. The computation is based on channel order statistics; here, we give an abstract formulation of the required probabilities, while their evaluation is detailed in Appendices~\ref{app:os}--\ref{app:mc}. 

The results follow by considering the error vectors covered in addition to those corrected by the \ac{BDD}.
First, we consider Chase-II decoding with test-pattern set $\Patchase$. It can be partitioned by Hamming weight as ${\Patchase = \bigcup_{w=0}^p \Pat_w}$, where
${\Pat_w = \left\{\one_{\calK} \;\middle\vert\; \calK\in \pws_{[p]},\, |\calK|=w\right\}}$.
We further define ${\Pat_{[0,w]} \triangleq \bigcup_{i\in[0,w]}\Pat_i}$, which corresponds to the test-pattern set of the restricted Chase algorithm with maximum pattern weight $w$.

For $\vrho\in\Pat_w$, the additionally covered error vectors have weight $\wham(\ve)=w+t$ and contain at most $t$ errors among the $n-p$ most reliable positions.
Hence,
\begin{equation*}
\Padd_{\Pat_{[0,w-1]}}(\Pat_w)
= P(\wham(\ve_{[p]}) \mathrel{\geq} w\mid \calE_{t+w})P(\calE_{t+w}).
\end{equation*}
Summing these contributions gives
\begin{equation}\label{eq:os_ler_chase}
P\chase(\vc\in\List)
= \Pbdd + \sum_{i=t+1}^{p+t}
P(\wham(\ve_{[p]}) \mathrel{\geq} i-t \mid \calE_i)P(\calE_i),
\end{equation}
in accordance with \cite{fossorier2002}.
For the restricted Chase algorithm with maximum pattern weight $\wmax$, this reduces to
\begin{equation}\label{eq:os_ler_res_chase}
P\reschase(\vc\in\List)
= \Pbdd + \sum_{i=t+1}^{\wmax+t}
P(\wham(\ve_{[p]}) \mathrel{\geq} i-t \mid \calE_i)P(\calE_i).
\end{equation}
Next, we derive the corresponding \ac{LER} expression for Chase-like decoding with logistic-weight test patterns $\Patlw$.

\begin{definition}[Logistic-Weight Patterns]\label{def:lw}
    We write $\patprime\prec\pat$ if $\patprime$ has smaller logistic weight than $\pat$, with ties broken first by Hamming weight and then by binary representation.
    The logistic-weight pattern set $\Patlw$ of cardinality $\npat$ consists of the $\npat$ smallest patterns with respect to $\prec$.
    For any $\pat\in\Patlw$, $\Patlwprec \triangleq \{\patprime\in\Patlw \mid \patprime\prec\pat\}$ denotes its preceding patterns.
\end{definition}
\begin{property}\label{property:subpatterns}
    If $\pat\in\Patlw$ and $\patprime\prec\pat$, then $\patprime\in\Patlw$.
\end{property}

\begin{remark}
    From \cref{property:subpatterns}, every pattern $\patprime$ whose support $\Ind_{\patprime}$ is a proper subset of that of $\pat\in\Patlw$ belongs to $\Patlwprec$, i.e.,
    \[
        \patprime\in\left\{\one_{\calK}\mid
        \calK\in\pws_{\Ind_{\pat}}\setminus\{\Ind_{\pat}\}\right\}
        \subseteq\Patlwprec.
    \]
    This enables an efficient sequential implementation of Chase-like decoding.
    In particular, a preceding pattern $\patprime\in\Patlwprec$ differing from $\pat$ in only one position $i$ has already been decoded.
    Its syndrome can therefore be reused and updated by the syndrome contribution of position $i$ only, avoiding a full syndrome computation.
\end{remark}

We next identify the smallest pattern with respect to $\prec$ that covers a given error vector $\ve$. This result is then used to derive the \ac{LER} for logistic-weight patterns.

\begin{lemma}[Minimal covering pattern]\label{lemma:min_lw}
    Let $\ve\in\Ftwo^n$ be an error vector with $\wham(\ve)=w$. The smallest pattern in $\Ball_t(\ve)$ with respect to $\prec$ is obtained by setting the $\min(t,w)$ highest-index nonzero positions of $\ve$ to zero. We call this the minimal covering
    pattern of $\ve$.
\end{lemma}
\begin{IEEEproof}
    Follows directly from the definition of the logistic-weight.    
\end{IEEEproof}

\begin{lemma}\label{lemma:add_vec}
Let $\Patlw$ be the logistic-weight pattern set from \cref{def:lw} and let $\pat\in\Patlw\setminus\{\zero\}$.
The error vectors covered by $\pat$ but by none of the preceding patterns $\Patlwprec$  are
\begin{equation*}
    \left\{
        \pat\oplus\one_{\calK}
        \;\middle|\;
        \calK\in\pws_{[\imax_{\pat}+1,n]},\,
        |\calK|=t
    \right\}.
\end{equation*}
For $\pat=\zero$, the additionally covered error vectors are $\Ball_t(\zero)$.
\end{lemma}

\begin{IEEEproof}
By \cref{lemma:min_lw}, the pattern $\pat$ newly covers exactly those error vectors for which it is the minimal covering pattern, which are precisely the vectors stated above.
\end{IEEEproof}

The probability that the transmitted codeword is contained in the list generated by Chase-like decoding with logistic-weight patterns is given by the following proposition.

    \begin{proposition}\label{prop:ler_lw_patterns}
    Let $\Patlw$ be the logistic-weight pattern set from \cref{def:lw}, with $\wmax=\max_{\pat\in\Patlw}\wham(\pat)$. Then, the probability that the transmitted codeword is in the list is
    \begin{equation}\label{eq:os_ler_lw}
    \begin{aligned}
        &P\lw(\vc\in\List) = \\
        &\Pbdd+ \sum_{i = 1}^{\wmax} \sum_{\pat\in\Pat_{i}}P(\ve_{[\imax_{\pat}]}=\pat_{[\imax_{\pat}]}\mid \calE_{t+i}) P(\calE_{t+i}),
    \end{aligned}
    \end{equation}
    with $\Pat_{w} = \{\pat\in\Patlw\;\vert\;\wham(\pat)=w\}$.
\end{proposition}

\begin{IEEEproof}
    The zero pattern covers $\Ball_t(\zero)$, contributing $\Pbdd$.
    By \cref{lemma:add_vec}, each $\pat\in\Pat_i$ additionally covers exactly those error vectors of weight $t+i$ that satisfy
    $\ve_{[\imax_{\pat}]}=\pat_{[\imax_{\pat}]}$.
    Summing these disjoint contributions over all $\pat\in\Patlw\setminus\{\zero\}$ yields \eqref{eq:os_ler_lw}.
\end{IEEEproof}
The conditional probabilities in \eqref{eq:os_ler_lw} can be calculated using order statistics, as detailed in Appendices~\ref{app:cond} and \ref{app:mc}.

\subsection{List Error Rate by Covered Space Probability}
Another approach to estimate the \ac{LER} of Chase-like algorithms uses the covered-space probability $\Pcov(\vy)$. Averaging over the channel output yields
\begin{equation*}
    \Pcl(\vc\in\List)
    = \mathbb{E}_{\vy}\left[\Pcov(\vy)\right]
    \approx \frac{1}{N}\sum_{i=1}^{N}
    \sum_{\vv\in\Vcov_i}\Pna(\vv\mid\vy_i),
\end{equation*}
which is accurate for a sufficiently large number $N$ of different channel outputs $\vy_i$. This method applies to arbitrary test-pattern sets $\Pat$. However, if computing overlaps between covered regions is too costly, directly estimating $\Pcl(\vc\in\List)$ by Monte Carlo simulation may be more efficient.

\section{Pattern Designs for Chase-like Decoding}
For a given number of patterns $\npat$, the optimal test-pattern set minimizes the \ac{LER}, i.e.,
\begin{equation*}
    \Pat_\mathrm{opt}
    \triangleq
    \argmin_{\substack{\Pat\subseteq\Ftwo^n:|\Pat|=\npat}}
    P_\Pat(\vc\notin\List).
\end{equation*}
Finding $\Pat_\mathrm{opt}$ is a maximum $\npat$-coverage problem, i.e., maximizing a submodular set function under a cardinality constraint, and is NP-hard \cite{nemhauser1978}. We therefore consider two iterative design algorithms: A greedy reference algorithm and a novel heuristic for finding good test-pattern sets. Both iteratively select patterns from a candidate set.

\subsection{Generating Candidate Patterns}\label{subsec:MLP}
In \cite{solomon2020}, an algorithm was introduced to enumerate all vectors in $\Ftwo^n$ according to a metric for a given realization $\vy$. We use this algorithm to generate candidate patterns in decreasing order of their metric, averaged over multiple realizations, and denote the resulting ordering by $\MLP$ (ordered candidate patterns), where $\MLP(i)$ is the pattern with the $i$-th largest average metric.

We distinguish two metrics: $\MLP_{\Ball_t(\vv)}$ orders patterns by $P(\ve\in\Ball_t(\pat))$, whereas $\MLP_{\vv}$ orders them by $P(\ve=\pat)$.

\subsection{Greedy Pattern Design}\label{subsec:greedy}
For the maximum $\npat$-coverage problem, the greedy algorithm that repeatedly selects the pattern with the largest marginal gain achieves at least a fraction $1-\nicefrac{1}{e}$ of the optimal covered probability \cite{nemhauser1978}, and in general this approximation ratio cannot be improved by any polynomial-time algorithm \cite{feige1998}. Thus, we use the greedy algorithm as a reference design.
Starting with $\Pat=\varnothing$, patterns are added until $|\Pat|=\npat$. At each step, the pattern with the largest reduction in \ac{LER} is selected as 
\begin{equation*}
\pat
= \argmin_{\vv\in\Ftwo^n} P_{\Pat\cup{\vv}}(\vc\notin\List)
= \argmax_{\vv\in\Ftwo^n} \Padd_{\Pat}(\vv).
\end{equation*}
The search can be accelerated using $\MLP_{\Ball_t(\vv)}$. After evaluating the first $j$ candidates, let $\pat_\mathrm{best}$ denote the best pattern found. Since
\begin{equation*}
    \Padd_{\Pat}(\pat_j)
    \leq P(\ve\in\Ball_t(\pat_j)),
    \qquad
    \pat_j=\MLP_{\Ball_t(\vv)}(j),
\end{equation*}
further candidates can be discarded once
\begin{equation*}
    P(\ve\in\Ball_t(\pat_j))
    \leq \Padd_{\Pat}(\pat_\mathrm{best}),
\end{equation*}
because $\MLP_{\Ball_t(\vv)}$ orders patterns by non-increasing covered probability.

The greedy guarantee should be interpreted with care. First, it bounds the covered probability rather than the \ac{LER} and is therefore loose at low \acp{LER}; for an optimal \ac{LER} of $10^{-4}$, it only guarantees an \ac{LER} below approximately $0.37$. Second, polynomial complexity in the number of candidates does not imply tractability here, since the candidate space has size $2^n$. We therefore propose a heuristic design with lower complexity. Moreover, the approximation guarantee does not preclude other heuristics, e.g., based on removing or exchanging patterns, from outperforming the greedy design on individual instances. We thus use the greedy design as a strong reference rather than an approximation of the optimum.

\subsection{Maximum Covering Pattern Design}\label{subsec:MCP}
A good test-pattern set should cover the most likely error vectors, i.e.,
${\min_{\pat\in\Pat}\dham(\pat,\MLP_{\vv}(i))\leq t}$ for small $i$.
Moreover, $P(\ve=\pat)$ can serve as a proxy for the marginal gain $\Padd_{\Pat}(\pat)$.

Based on these observations, we propose the \acf{MCOC} design in Algorithm~\ref{alg:os_cov}.
If an error vector $\ve$ is not yet covered, the most likely pattern covering $\ve$,
\[
    \pat=\MLP_{\vv}(j),
    \qquad
    j=\min\{i\mid \ve\in\Ball_t(\MLP_{\vv}(i))\},
\]
is added to $\Pat$.
Similar to logistic-weight patterns, \ac{MCOC} supports efficient sequential decoding.

\begin{remark}
Each pattern added by \ac{MCOC} is within Hamming distance $t$ of a previously added pattern since it is already covered. Hence, its syndrome can be obtained from an already computed syndrome using at most $t$ exclusive-or operations. For logistic-weight patterns, only one exclusive-or operation is required.
\end{remark}

\begin{algorithm}[htb]
    \caption{Maximally Covering Ordered Candidate Patterns.} 
    \label{alg:os_cov}
    \Input{\Ac{SNR} $E_\mathrm{b}/N_0$, pattern number $q$, parameters $(n,k,t)$.}
    \Output{Set of patterns $\Pat$.}
    \algrule
    $\Pat \gets \varnothing\text{;} \;i\gets 1$\;
    \While{$|\Pat|<q$}{
        $\pat \gets \MLP_{\vv}(i)$\tcp*{next likely pattern}
        \If{$\min_{\patprime \in \Pat} \dham (\pat, \patprime) > t$}{
            \For{$j \gets 1$ \KwTo $i$}{
                $\pat_j = \MLP_{\vv}(j)$\;
                \If{$\dham(\pat,\pat_j)\leq t$}{
                    $\Pat\gets \Pat\cup\{\pat_j\}$\;
                    \textbf{break}\tcp*{go to line 13}
                }
            }
        }    
        $i\gets i+1$\;
    }
	\Return $\Pat$
\end{algorithm}

\section{Results}

In \cref{fig:Ham_128}, direct Monte Carlo simulation, order-statistics calculations, and the covered-space method yield the same \ac{LER} for the $(128,120)$ extended Hamming code with $128$ logistic-weight patterns. This agreement, also verified for high-rate \ac{BCH} codes with $1\leq t\leq3$, validates all three methods for \ac{LER} evaluation.

Additionally, \cref{fig:Ham_128} compares four test-pattern sets with $\npat=16$ and $64$. For $\npat=16$ ($p=4$), Chase-II performs slightly worse than logistic-weight, greedy, and \ac{MCOC} patterns, which perform similarly. For $\npat=64$, Chase-II with $p=6$ exhibits a substantially higher \ac{LER}. The greedy design slightly outperforms \ac{MCOC} at high \ac{SNR}, while logistic-weight patterns are only $0.1\,\mathrm{dB}$ worse than the greedy design at an \ac{LER} of $10^{-4}$. Since the \ac{LER} of all considered designs lies below the \ac{ML} \ac{BLER}, they achieve near-\ac{ML} performance according to \eqref{eq:bler_bound}. The \ac{MCOC} test-pattern set was designed at an \ac{SNR} of $6.0\,\mathrm{dB}$.

We further consider the $(256,239)$ extended \ac{BCH} code used as a component code in oFEC \cite{ofec2023}. Its \ac{BLER} for different test-pattern sets is shown in \cref{fig:bch_256_bler}. The \ac{MCOC} design with only $\npat=42$ outperforms restricted Chase with $\npat=93$, while $93$ logistic-weight patterns achieve nearly the same performance as \ac{MCOC}. Since restricted Chase with $\npat=93$, $p=8$, and $\wmax=3$ meets the target performance for oFEC \cite{wang2023}, these results indicate that \ac{MCOC} and logistic-weight patterns can substantially reduce the number of required \ac{BDD} attempts. A greedy design could not be obtained within reasonable computation time. The \ac{MCOC} patterns were again designed at an \ac{SNR} of $6.0\,\mathrm{dB}$.

The covering-code-based construction of \cite{tokushige2003}, using a $(7,4)$ Hamming code as an initialization to design $\npat=42$ and $93$ patterns, performs worse than both logistic-weight and \ac{MCOC} patterns and is therefore omitted from \cref{fig:bch_256_bler}. This suggests that the covering-code heuristic may be less effective for high-rate codes.

Finally, \cref{fig:bch_256_ler_snr} shows the \ac{LER} of the same \ac{BCH} code at a fixed \ac{SNR} of $5.5\,\mathrm{dB}$ as a function of the number of test patterns. As the number of patterns increases, the performance advantage of \ac{MCOC} and logistic-weight patterns over the Chase variants becomes more pronounced.

\begin{figure}[t]
    \centering
    \begin{tikzpicture}[spy using outlines={circle}]
\begin{writereplace}
\begin{axis}[%
    xshift=1.5cm,
    xmin=3,
    xmax=7.5,
    ymode=log,
    ymin=1e-5,
    ymax=1,
    xmajorgrids,
    xminorgrids=true,
    ymajorgrids,
    yminorticks=true,
    yminorgrids=true,
    major grid style={hellgrau},
    minor grid style={densely dotted, hellgrau},
    width=0.985\linewidth,
    height=\linewidth,
    tick align=outside,
    tick pos=left,
    xlabel={$E_\mathrm{b}/N_0$ in $\mathrm{dB}$},
    ylabel={LER},
    label style={font=\small},
    ticklabel style = {font=\small},
    every axis plot/.style={
    line width=1pt,
    mark size=2.25pt
    }
]

\addplot [\colorGreedy, \lineStyleA, mark=\markGreedy, mark options={solid}, legend image post style={solid}]
 table[col sep=comma]{%
2.50,5.55432E-01
3.00,5.55432E-01
3.50,3.46575E-01
4.00,1.74264E-01
4.50,6.72811E-02
5.00,1.95100E-02
5.50,4.15737E-03
6.00,5.91579E-04
6.50,6.5625E-05
7.00,4.51754E-06

};
\label{plot:ham128 greedy}

\addplot [\colorGreedy, \lineStyleB, mark=\markGreedy, mark options={solid}, forget plot]
 table[col sep=comma]{%
2.50,3.96156E-01
3.00,3.96156E-01
3.50,2.03781E-01
4.00,7.96716E-02
4.50,2.24847E-02
5.00,4.41316E-03
5.50,5.92632E-04
6.00,4.47368E-05
7.00,9.56938E-08

};

\addplot[\colorChase, \lineStyleA, mark=\markChase, mark options={solid}, legend image post style={solid}]
 table[col sep=comma]{%
2.50,7.51993E-01
3.00,5.66621E-01
3.50,3.60667E-01
4.00,1.86633E-01
4.50,7.56458E-02
5.00,2.30821E-02
5.50,5.36895E-03
6.00,8.68421E-04
6.50,1.00526E-04
7.00,8.27068E-06
7.50,4.27899E-07
}; 
\label{plot:ham128 chase}

\addplot [\colorChase, \lineStyleB, mark=\markChase, mark options={solid}, forget plot]
 table[col sep=comma]{%
2.50,6.37610E-01
3.00,4.27815E-01
3.50,2.33856E-01
4.00,1.00308E-01
4.50,3.27532E-02
5.00,7.95632E-03
5.50,1.38000E-03
6.00,1.71579E-04
6.50,1.82456E-05
7.00,1.07411E-06
7.50,3.80561E-08
};

\addplot [\colorLW, \lineStyleA, mark=\markLW, mark options={solid}, legend image post style={solid}]
  table[col sep=comma]{%
2.50,7.46439E-01
3.00,5.57486E-01
3.50,3.50701E-01
4.00,1.76282E-01
4.50,6.87432E-02
5.00,2.02158E-02
5.50,4.27737E-03
6.00,5.95789E-04
6.50,7.52632E-05
7.00,4.38596E-06
7.50,2.41429E-07
}; 
\label{plot:ham128 lw}

\addplot [\colorLW, \lineStyleB, mark=\markLW, mark options={solid}, forget plot]
  table[col sep=comma]{%
2.50,6.14878E-01
3.00,3.99691E-01
3.50,2.06868E-01
4.00,8.12053E-02
4.50,2.33584E-02
5.00,4.68263E-03
5.50,6.50000E-04
6.00,6.89474E-05
6.50,2.95322E-06
};

\addplot [\colorMCOC, \lineStyleA, mark=\markMCOC, mark options={solid}, legend image post style={solid}]
  table[col sep=comma]{

3.00,5.54889E-01
3.50,3.46663E-01
4.00,1.74321E-01
4.50,6.70684E-02
5.00,1.97316E-02
5.50,4.00526E-03
6.00,5.47368E-04
6.50,5.78947E-05
7.00,3.62976E-06

}; 
\label{plot:ham128 mcoc}

\addplot [\colorMCOC, \lineStyleB, mark=\markMCOC, mark options={solid}, forget plot]
  table[col sep=comma]{

3.00,3.94842E-01
3.50,2.04921E-01
4.00,7.97684E-02
4.50,2.23105E-02
5.00,4.76842E-03
5.50,6.36842E-04
6.00,5.90643E-05
6.50,2.64480E-06
7.00,9.39179E-08
};

\addplot [black, \lineStyleC, mark=\markML,  mark options={solid}] 
table[col sep=comma]{
2.50,7.82982E-01
3.00,6.05557E-01
3.50,3.97797E-01
4.00,2.14727E-01
4.50,9.33432E-02
5.00,3.15111E-02
5.50,8.42263E-03
6.00,1.75526E-03
6.50,2.91053E-04
7.00,3.31579E-05
7.50,3.31269E-06
};
\label{plot:ham128 ml bler}

\addplot [violettblau, loosely dotted, mark=o, mark options={solid}, forget plot]
  table[col sep=comma]{
2.50,5.45586E-01
3.00,3.27381E-01
3.50,1.52704E-01
4.00,5.30384E-02
4.50,1.31253E-02
5.00,2.18368E-03
5.50,2.56316E-04
6.00,1.92982E-05
6.50,8.92061E-07
}; 
\label{plot:ham128 lw 128 dmc}

 \addplot [violettblau!60, loosely dotted, dash phase=1.2pt, mark=+,  mark options={solid}, forget plot, legend image post style={dash phase=0pt}]
 table[col sep=comma]{%
2.50,5.45075E-01
3.00,3.27531E-01
3.50,1.53068E-01
4.00,5.29890E-02
4.50,1.30240E-02
5.00,2.18198E-03
5.50,2.41819E-04
6.00,1.72539E-05
6.50,7.90383E-07
7.00,2.28063E-08
};
\label{plot:ham128 lw 128 cs}

\addplot [violettblau!30, loosely dotted, mark=x, dash phase=2.4pt,  mark options={solid}, forget plot, legend image post style={dash phase=0pt}]
  table[col sep=comma]{%
2.50,5.45062E-01
3.00,3.27511E-01
3.50,1.53075E-01
4.00,5.30068E-02
4.50,1.30162E-02
5.00,2.17583E-03
5.50,2.43312E-04
6.00,1.71794E-05
6.50,8.72473E-07
}; 
\label{plot:ham128 lw 128 os}

\addplot [black, \lineStyleA] 
table[col sep=comma]{%
  2,2
  2,3
}; 
\label{plot:ham128 q=16}

\addplot [black, \lineStyleB] 
table[col sep=comma]{%
  2,2
  2,3
}; 
\label{plot:ham128 q=64}

\addplot [black, loosely dotted] 
table[col sep=comma]{%
  2,2
  2,3
}; 
\label{plot:ham128 q=128}

\coordinate (coordSpyA) at (axis cs:5.00,2.18198E-03);
\coordinate (coordMagA) at (axis cs:3.50,1E-03);

\coordinate (coordA) at (axis description cs:1,1);
\coordinate (coordB) at (axis description cs:0,0);

\spy[magnification=2.75,
    size=1cm,
    connect spies,
    every spy in node/.append style={fill=white},
    spy connection path={
        \fill[black]
            ($(tikzspyonnode.91)+(0,-0.1pt)$) --
            ($(tikzspyonnode.91)+(0,0.1pt)$) --
            ($(tikzspyinnode.91)+(0,0.4pt)$) --
            ($(tikzspyinnode.91)+(0,-0.4pt)$) --
            cycle;
        \fill[black]
            ($(tikzspyonnode.290)+(-0.3pt,0)$) --
            ($(tikzspyonnode.290)+(0.3pt,0)$) --
            ($(tikzspyinnode.290)+(0,-0.4pt)$) --
            ($(tikzspyinnode.290)+(0,0.4pt)$) --
            cycle;
    }
] on (coordSpyA) in node at (coordMagA);
\end{axis}

\node (tabA) [
    shape=rectangle,
    anchor=north east,
    font=\footnotesize,
    fill=white,
    fill opacity=0.7,
    text opacity=1,
] at (coordA) {
    \begin{NiceTabular}{@{}l@{}}
        Type\\
        \midrule
        \ref{plot:ham128 chase} Chase\\
        \ref{plot:ham128 lw} LW\\
        \ref{plot:ham128 mcoc} MCOC\\
        \ref{plot:ham128 greedy} Greedy\\[4pt]
        Count\\
        \midrule
        \ref{plot:ham128 q=16} $q=\phantom{0}16$\\
        \ref{plot:ham128 q=64} $q=\phantom{0}64$\\
        \ref{plot:ham128 q=128} $q=128$\\[4pt]
        BLER\\
        \midrule
        \ref{plot:ham128 ml bler} ML\\
    \end{NiceTabular}
};

\node (tabB) [
    shape=rectangle,
    anchor=south west,
    font=\footnotesize,
    fill=white,
    fill opacity=0.7,
    text opacity=1,
] at (coordB) {
    \begin{NiceTabular}{@{}l@{}}
        LER method:
        LW, $q=128$\\
        \midrule
        \ref{plot:ham128 lw 128 os} Order Statistics\\
        \ref{plot:ham128 lw 128 cs} Covered Space\\
        \ref{plot:ham128 lw 128 dmc} Monte Carlo\\
    \end{NiceTabular}
};
\end{writereplace}
\end{tikzpicture}  
    \caption{\Acl{LER} of the (128,120) extended Hamming code decoded with different sets of patterns of size $q\in\{16,64,128\}$. }
    \label{fig:Ham_128}
\end{figure}
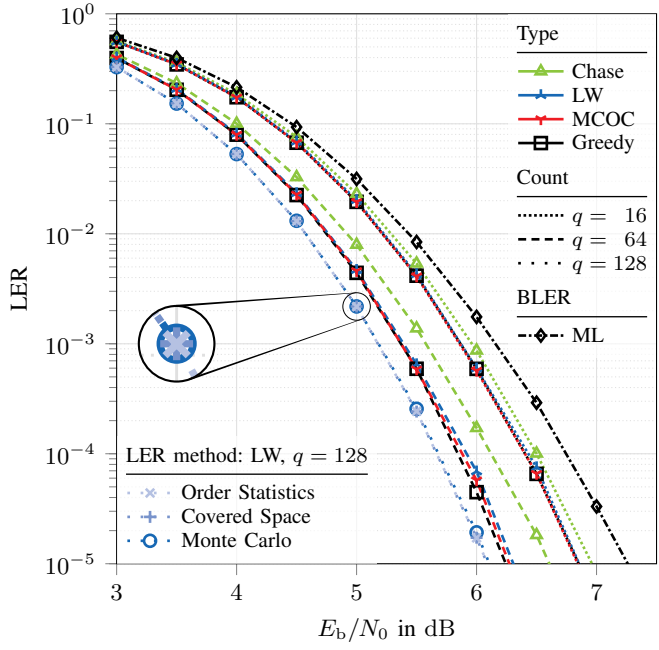

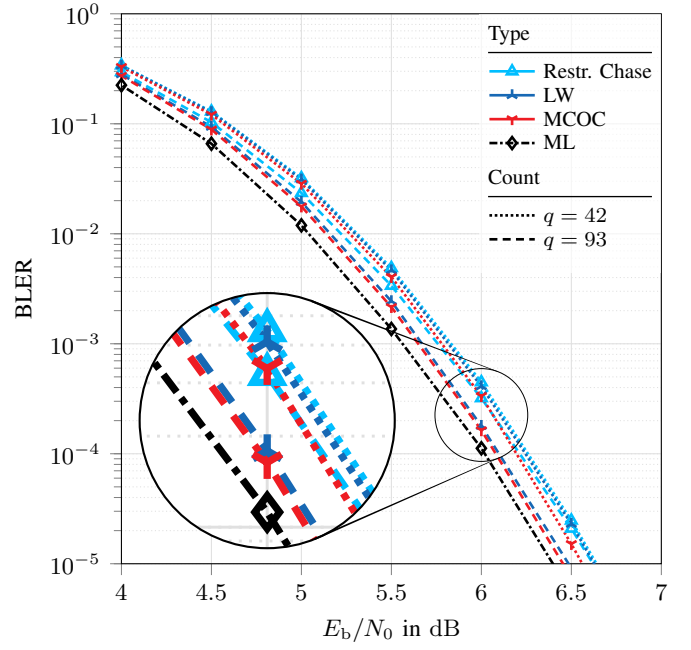
\begin{figure}[t]
    \centering
    \begin{tikzpicture}[spy using outlines={circle}]
\begin{writereplace}
\begin{axis}[%
        xshift=1.5cm,
        xmin=4.0,
        xmax=6.75,
        ymode=log,
        ymin=1e-5,
        ymax=01,
        xmajorgrids,
		xminorgrids=true,
		ymajorgrids,
		yminorticks=true,
		yminorgrids=true,
		major grid style={hellgrau},
    	minor grid style={densely dotted, hellgrau},
        width=0.985\linewidth,
        height=\linewidth,
		tick align=outside,
    	tick pos=left,
        xlabel={$E_\mathrm{b}/N_0$ in $\mathrm{dB}$},
        ylabel={BLER},
        label style={font=\small},
		ticklabel style = {font=\small},
		every axis plot/.style={
        line width=1pt,
        mark size=2.25pt
    	}
]

\addplot [color=\colorRC, \lineStyleA,  mark=\markRC, mark options={solid}, legend image post style={solid}]
 table[col sep=comma]{%
2.50,9.53844E-01
3.00,8.40397E-01
3.50,6.17070E-01
4.00,3.39979E-01
4.50,1.29457E-01
5.00,3.20663E-02
5.50,4.88684E-03
6.00,4.47368E-04
6.50,2.47368E-05
7.00,9.66856E-07
}; 
\label{plot:bch256 rc 42}

\addplot[color=\colorRC, \lineStyleB, mark=\markRC, mark options={solid}]
table[col sep=comma]{%
2.50,9.43869E-01
3.00,8.12712E-01
3.50,5.71405E-01
4.00,2.95372E-01
4.50,1.03311E-01
5.00,2.35189E-02
5.50,3.37105E-03
6.00,3.20526E-04
6.50,2.08772E-05
7.00,9.75686E-07
}; 
\label{plot:bch256 rc 93}

\addplot [color=\colorLW, \lineStyleA,  mark=\markLW, mark options={solid}, legend image post style={solid}]
  table[col sep=comma]{%
2.50,9.53609E-01
3.00,8.39272E-01
3.50,6.14518E-01
4.00,3.37190E-01
4.50,1.27545E-01
5.00,3.11005E-02
5.50,4.64526E-03
6.00,4.10000E-04
6.50,2.29825E-05
7.00,8.75120E-07
}; 
\label{plot:bch256 lw 42}

\addplot [color=\colorLW, \lineStyleB, mark=\markLW, mark options={solid}]
 table[col sep=comma]{%
2.50,9.41482E-01
3.00,8.06558E-01
3.50,5.58917E-01
4.00,2.81128E-01
4.50,9.37137E-02
5.00,1.94658E-02
5.50,2.43474E-03
6.00,1.78947E-04
6.50,9.12281E-06
};
\label{plot:bch256 lw 93}

\addplot [color=\colorMCOC, \lineStyleA,  mark=\markMCOC, mark options={solid}, legend image post style={solid}]
  table[col sep=comma]{
3.00,8.36444E-01
3.50,6.08272E-01
4.00,3.29198E-01
4.50,1.21246E-01
5.00,2.85605E-02
5.50,4.04842E-03
6.00,3.34737E-04
6.50,1.51316E-05
7.00,5.01253E-07

}; 
\label{plot:bch256 mcoc 42}

\addplot[color=\colorMCOC, \lineStyleB,  mark=\markMCOC, mark options={solid}]
 table[col sep = comma]{%
3.00,8.02157E-01
3.50,5.51620E-01
4.00,2.73984E-01
4.50,8.93958E-02
5.00,1.78526E-02
5.50,2.10947E-03
6.00,1.48947E-04
6.50,6.71053E-06
7.00,2.49439E-07
};
\label{plot:bch256 mcoc 93}

\addplot [color=\colorMCOC, \lineStyleA,  mark=\markMCOC, mark options={solid}, legend image post style={solid}]
  table[col sep=comma]{
3.00,8.36319E-01
3.50,6.08326E-01
4.00,3.30219E-01
4.50,1.21249E-01
5.00,2.83600E-02
5.50,4.05263E-03
6.00,3.31579E-04
6.50,1.82456E-05
7.00,5.26316E-07
};

\addplot [black, \lineStyleC, mark=\markML,  mark options={solid}] 
table[col sep=comma]{
2.50,9.25363E-01
3.00,7.64652E-01
3.50,4.95329E-01
4.00,2.24238E-01
4.50,6.58821E-02
5.00,1.19216E-02
5.50,1.36316E-03
6.00,1.12105E-04
6.50,5.36842E-06

};
\label{plot:bch256 ml}

\coordinate (coordA) at (axis description cs:1,1);

\coordinate (coordSpyB) at (axis cs:6,2.25E-04);
\coordinate (coordMagB) at (axis cs:4.81,2E-04);
\end{axis}

\spy[
    magnification=2.75,
    size=3.375cm,
    connect spies,
    every spy in node/.append style={fill=white},
    spy connection path={
        \fill[black]
            ($(tikzspyonnode.69)+(-0.3pt,0)$) --
            ($(tikzspyonnode.69)+(0.3pt,0)$) --
            ($(tikzspyinnode.70)+(0,0.4pt)$) --
            ($(tikzspyinnode.70)+(0,-0.4pt)$) --
            cycle;
        \fill[black]
            ($(tikzspyonnode.290)+(-0.3pt,0)$) --
            ($(tikzspyonnode.290)+(0.3pt,0)$) --
            ($(tikzspyinnode.292)+(0,-0.4pt)$) --
            ($(tikzspyinnode.292)+(0,0.4pt)$) --
            cycle;
    }
] on (coordSpyB) in node at (coordMagB);

\node (tabA) [
    shape=rectangle,
    anchor=north east,
    font=\footnotesize,
    fill=white,
    fill opacity=0.7,
    text opacity=1,
] at (coordA) {
    \begin{NiceTabular}{@{}l@{}}
        Type\\
        \midrule
        \ref{plot:bch256 rc 42} Restr. Chase\\
        \ref{plot:bch256 lw 42} LW\\
        \ref{plot:bch256 mcoc 42} MCOC\\
        \ref{plot:ham128 ml bler} ML\\[4pt]
        Count\\
        \midrule
        \ref{plot:ham128 q=16} $q=42$\\
        \ref{plot:ham128 q=64} $q=93$\\[4pt]
    \end{NiceTabular}
};
\end{writereplace}
\end{tikzpicture}
    \caption{\Acl{BLER} of the (256,239) extended BCH code decoded with $q\in\{42,93\}$ patterns for different pattern sets. The restricted Chase pattern sets are with $p\in\{6,8\}$ and $\wmax = 3$.}
    \label{fig:bch_256_bler}
\end{figure}

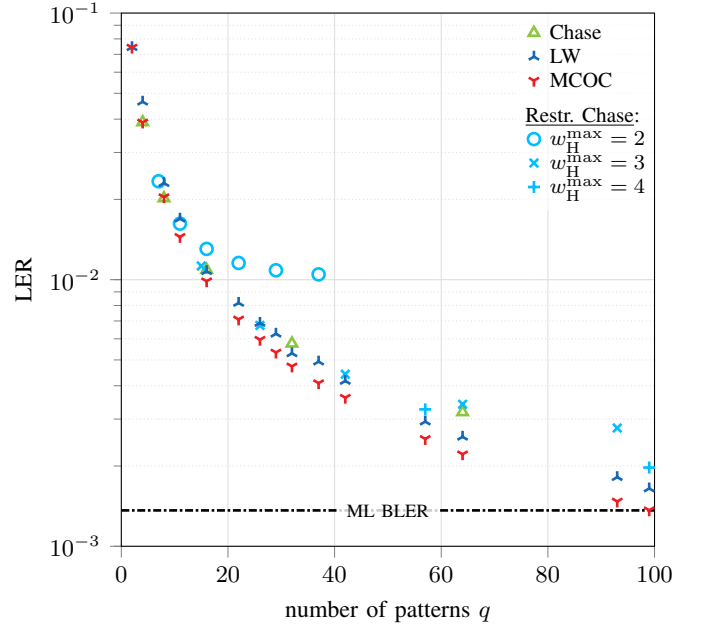
\begin{figure}[hbt!]
    \centering
    \begin{tikzpicture}
\begin{axis}[%
    xshift=1.5cm,
    xmin=0,
    xmax=100,
    ymode=log,
    ymin=1e-3,
    ymax=0.1,
    xmajorgrids,
    xminorgrids=true,
    ymajorgrids,
    yminorticks=true,
    yminorgrids=true,
    major grid style={hellgrau},
    minor grid style={densely dotted, hellgrau},
    width=0.975*\columnwidth,
    height=0.975*\columnwidth,
    tick align=outside,
    tick pos=left,
    xlabel={number of patterns $\npat$},
    ylabel={LER},
    label style={font=\small},
    ticklabel style = {font=\small},
    every axis plot/.style={
        line width=1pt,
        mark size=2.25pt
    }
]

\addplot [\colorChase, only marks, mark=\markChase, mark options={solid}]
 table[col sep=comma]{%
4.0,0.0389753
8.0,0.0202126
16.0,0.0108874
32.0,0.00576263
64.0,0.00318632
128.0,0.00178421
}; 
\label{plot:npatterns chase}

\addplot [\colorRC, only marks, mark=o, mark options={solid}]
  table[col sep=comma]{%
7.0,0.0233926
11.0,0.0162242
16.0,0.0130432
22.0,0.0115532
29.0,0.0108453
37.0,0.0104695
}; 
\label{plot:npatterns rc w=2}

\addplot [\colorRC, only marks, mark=x, mark options={solid}]
  table[col sep=comma]{%
15.0,0.0112547
26.0,0.00673526
42.0,0.00442
64, 3.40842E-03
93.0,0.00277789
}; 
\label{plot:npatterns rc w=3}

\addplot [\colorRC, only marks, mark=+, mark options={solid}]
  table[col sep=comma]{%
57.0,0.00326263
99.0,0.00197421
}; 
\label{plot:npatterns rc w=4}

\addplot [\colorLW, only marks, mark=\markLW, mark options={solid}]
 table[col sep=comma]{%
2.0,0.0744537
4.0,0.0464858
8.0,0.0230168
11.0,0.0169284
16.0,0.01076
22.0,0.00818
26.0,0.00688
29.0,0.00626632
32.0,0.00529158
37.0,0.00492632
42.0,0.00416947
57.0,0.00293421
64.0,0.00257368
93.0,0.00181632
99.0,0.00164737
};
\label{plot:npatterns lw}

\addplot [\colorMCOC, only marks, mark=\markMCOC, mark options={solid}, legend image post style={solid}]
  table[col sep=comma]{
2.0,0.0742409
4.0,0.0388804
8.0,0.0204013
11.0,0.0144765
16.0,0.00986565
22.0,0.00709652
26.0,0.00595435
29.0,0.00533174
32.0,0.00472957
37.0,0.00409043
42.0,0.00360826
57.0,0.00252696
64.0,0.00221391
93.0,0.00147261
99.0,0.00136478
}; 
\label{plot:npatterns mcoc}

\addplot[mark=none, \lineStyleC, black, samples = 2,domain=0:100] {1.36316E-03};
\label{plot:npatterns ml bler}

\node[fill=white, fill opacity = 0.8, text opacity = 1]  at (axis cs: 50, 1.36e-3){\scriptsize ML BLER};

\coordinate (coordA) at (axis description cs:1,1);
\coordinate (coordB) at (axis description cs:1,1);
\end{axis}

\node (tabA) [
    shape=rectangle,
    anchor=north east,
    font=\footnotesize,
    fill=white,
    fill opacity=0.7,
    text opacity=1,
] at (coordA) {
    \begin{NiceTabular}{@{}l@{}}
        \parbox[b]{0.22cm}{\hfil\ref{plot:npatterns chase}\hfil} Chase\\
        \parbox[b]{0.22cm}{\hfil\ref{plot:npatterns lw}\hfil} LW\\
        \parbox[b]{0.22cm}{\hfil\ref{plot:npatterns mcoc}\hfil} MCOC\\[4pt]
        \underline{Restr. Chase}:\\
        \parbox[b]{0.22cm}{\hfil\ref{plot:npatterns rc w=2}\hfil} $\wmax = 2$\\
        \parbox[b]{0.22cm}{\hfil\ref{plot:npatterns rc w=3}\hfil} $\wmax = 3$\\
        \parbox[b]{0.22cm}{\hfil\ref{plot:npatterns rc w=4}\hfil} $\wmax = 4$\\
    \end{NiceTabular}
};

\end{tikzpicture}  
    \caption{\Acl{LER} over the number of patterns of different sets used to decode the (256,239) extended BCH code at an \ac{SNR} of $5.5\,\mathrm{dB}$.}
    \label{fig:bch_256_ler_snr}
\end{figure}

\section{Conclusion}
We presented three methods for evaluating test-pattern sets for Chase-like decoding, all providing accurate estimates of the \acl{LER}. In addition, we proposed a design algorithm that achieves a given performance target with substantially fewer test patterns than conventional Chase-based sets while performing similarly to the greedy design. Future work may extend these insights to the design of test-pattern sets for iterative decoding of product codes and their generalizations.

\section*{Acknowledgments}
The authors would like to thank Benjamin Castellaz for pointing us to the methods for efficient order-statistics calculations and the anonymous reviewers for their valuable feedback.

\newpage

\IEEEtriggeratref{15} 
\bibliographystyle{IEEEtran}
\bibliography{bibliofile}

\begin{thebibliography}{10}
\providecommand{\url}[1]{#1}
\csname url@samestyle\endcsname
\providecommand{\newblock}{\relax}
\providecommand{\bibinfo}[2]{#2}
\providecommand{\BIBentrySTDinterwordspacing}{\spaceskip=0pt\relax}
\providecommand{\BIBentryALTinterwordstretchfactor}{4}
\providecommand{\BIBentryALTinterwordspacing}{\spaceskip=\fontdimen2\font plus
\BIBentryALTinterwordstretchfactor\fontdimen3\font minus
  \fontdimen4\font\relax}
\providecommand{\BIBforeignlanguage}[2]{{%
\expandafter\ifx\csname l@#1\endcsname\relax
\typeout{** WARNING: IEEEtran.bst: No hyphenation pattern has been}%
\typeout{** loaded for the language `#1'. Using the pattern for}%
\typeout{** the default language instead.}%
\else
\language=\csname l@#1\endcsname
\fi
#2}}
\providecommand{\BIBdecl}{\relax}
\BIBdecl

\bibitem{hocquenghem1959}
A.~Hocquenghem, ``\BIBforeignlanguage{French}{Codes correcteurs d'erreurs},''
  \emph{\BIBforeignlanguage{French}{Chiffres}}, vol.~2, pp. 147--156, Sep.
  1959.

\bibitem{bose1960}
R.~C. Bose and D.~K. {Ray-Chaudhuri}, ``On a class of error correcting binary
  group codes,'' \emph{Information and Control}, vol.~3, no.~1, pp. 68--79,
  Mar. 1960.

\bibitem{lin2004}
S.~Lin and D.~J. Costello, \emph{Error {{Control Coding}}: {{Fundamentals}} and
  {{Applications}}}, 2nd~ed.\hskip 1em plus 0.5em minus 0.4em\relax
  Pearson-Prentice Hall, 2004.

\bibitem{forney1966a}
G.~D. {Forney Jr.}, \emph{Concatenated Codes}, ser. Research Monograph.\hskip
  1em plus 0.5em minus 0.4em\relax Cambridge, MA, USA: MIT Press, 1966, no.~37.

\bibitem{matuz2025}
B.~Matuz, E.~B. Yacoub, and S.~Calabr{\`o}, ``Serially concatenated codes for
  data center networks,'' in \emph{2025 13th International Symposium on Topics
  in Coding (ISTC)}, Aug. 2025, pp. 1--5.

\bibitem{smith2012a}
B.~P. Smith, A.~Farhood, A.~Hunt, F.~R. Kschischang, and J.~Lodge, ``Staircase
  {{Codes}}: {{FEC}} for 100 {{Gb}}/s {{OTN}},'' \emph{J. Light. Technol.},
  vol.~30, no.~1, pp. 110--117, Jan. 2012.

\bibitem{sukmadji2022}
A.~Y. Sukmadji, U.~{Mart{\'i}nez-Pe{\~n}as}, and F.~R. Kschischang, ``Zipper
  codes,'' \emph{J. Light. Technol.}, vol.~40, no.~19, pp. 6397--6407, Oct.
  2022.

\bibitem{shehadeh2025}
M.~Shehadeh, F.~R. Kschischang, A.~Y. Sukmadji, and W.~Kingsford,
  ``Higher-order staircase codes,'' \emph{IEEE Trans. Inf. Theory}, vol.~71,
  no.~4, pp. 2517--2538, Apr. 2025.

\bibitem{ofec2023}
M.~A. Sluyski, ``Open {{ROADM MSA}} {{6.0 W B400G}} port digital specification
  ({{400G-800G}}),'' Dec. 2023.

\bibitem{pyndiah1998}
R.~Pyndiah, ``Near-optimum decoding of product codes: block turbo codes,''
  \emph{IEEE Trans. Commun.}, vol.~46, pp. 1003--1010, Aug. 1998.

\bibitem{chase1972}
D.~Chase, ``Class of algorithms for decoding block codes with channel
  measurement information,'' \emph{IEEE Trans. Inf. Theory}, vol.~18, no.~1,
  pp. 170--182, Jan. 1972.

\bibitem{duffy2022}
K.~R. Duffy, W.~An, and M.~M{\'e}dard, ``Ordered reliability bits guessing
  random additive noise decoding,'' \emph{IEEE Trans. Signal Process.},
  vol.~70, pp. 4528--4542, 2022.

\bibitem{shen2025}
Y.~Shen, W.~Song, L.~D. Blanc, Y.~Ren, A.~{Balatsoukas-Stimming}, A.~Alvarado,
  and A.~Burg, ``Iterative logistic weight based chase decoder for open forward
  error correction,'' in \emph{2025 Optical Fiber Communications Conference and
  Exhibition (OFC)}, Mar. 2025, pp. 1--3.

\bibitem{tokushige2003}
H.~Tokushige, T.~Koumoto, M.~P.~C. Fossorier, and T.~Kasami, ``Selection method
  of test patterns in soft-decision iterative bounded distance decoding
  algorithms,'' \emph{IEICE Trans. Fundam.}, vol. E86-A, no.~10, pp.
  2445--2451, Oct. 2003.

\bibitem{tokushige2010}
H.~Tokushige, M.~P. C.~Fossorier, and T.~Kasami, ``A test pattern selection
  method for a joint bounded-distance and encoding-based decoding algorithm of
  binary codes,'' \emph{IEEE Trans. Commun.}, vol.~58, no.~6, pp. 1601--1604,
  Jun. 2010.

\bibitem{cohen1997}
G.~Cohen, I.~Honkala, S.~Litsyn, and A.~Lobstein, \emph{Covering Codes}, ser.
  North-Holland Mathematical Library.\hskip 1em plus 0.5em minus 0.4em\relax
  Amsterdam, The Netherlands: Elsevier, 1997, vol.~54.

\bibitem{cardinale2026}
A.~Cardinale, W.~Song, B.~Chen, A.~Alvarado, A.~Burg, and Y.~Shen, ``Maximum
  coverage {{Chase}} decoder for optical interconnects,'' 2026.

\bibitem{forney1966}
G.~D. {Forney Jr.}, ``Generalized minimum distance decoding,'' \emph{IEEE
  Trans. Inf. Theory}, vol.~12, no.~2, pp. 125--131, Apr. 1966.

\bibitem{agrawal2000}
D.~Agrawal and A.~Vardy, ``Generalized minimum distance decoding in
  {{Euclidean}} space: performance analysis,'' \emph{IEEE Trans. Inf. Theory},
  vol.~46, no.~1, pp. 60--83, Jan. 2000.

\bibitem{fossorier2002}
M.~P.~C. Fossorier and S.~Lin, ``Error performance analysis for
  reliability-based decoding algorithms,'' \emph{IEEE Trans. Inf. Theory},
  vol.~48, no.~1, pp. 287--293, Jan. 2002.

\bibitem{he2025}
X.~He, L.~Chen, and Y.~Wu, ``Performance analysis and enhanced chase decoding
  of {{GII-BCH}} codes,'' \emph{IEEE Trans. Commun.}, vol.~73, no.~10, pp.
  8647--8658, Oct. 2025.

\bibitem{janz2025}
T.~Janz, S.~Oberm{\"u}ller, A.~Zunker, and S.~{ten Brink}, ``Soft-output from
  covered space decoding of product codes,'' in \emph{2025 13th International
  Symposium on Topics in Coding (ISTC)}, Aug. 2025, pp. 1--5.

\bibitem{nemhauser1978}
G.~L. Nemhauser, L.~A. Wolsey, and M.~L. Fisher, ``An analysis of
  approximations for maximizing submodular set functions---{{I}},''
  \emph{Mathematical Programming}, vol.~14, no.~1, pp. 265--294, Dec. 1978.

\bibitem{solomon2020}
A.~Solomon, K.~R. Duffy, and M.~M{\'e}dard, ``Soft maximum likelihood decoding
  using {{GRAND}},'' in \emph{2020 IEEE International Conference on
  Communications (ICC)}, Jun. 2020, pp. 1--6.

\bibitem{feige1998}
U.~Feige, ``A threshold of ln n for approximating set cover,'' \emph{J. ACM},
  vol.~45, no.~4, p. 634–652, Jul. 1998.

\bibitem{wang2023}
W.~Wang, Z.~Long, W.~Qian, K.~Tao, Z.~Wei, S.~Zhang, Z.~Feng, Y.~Xia, and
  Y.~Chen, ``Real-time {{FPGA}} investigation of potential {{FEC}} schemes for
  {{800G-ZR}}/{{ZR}}+ forward error correction,'' \emph{J. Light. Technol.},
  vol.~41, no.~3, pp. 926--933, Feb. 2023.

\bibitem{david2003}
H.~A. David and H.~N. Nagaraja, \emph{Order Statistics}, 3rd~ed.\hskip 1em plus
  0.5em minus 0.4em\relax Hoboken, NJ, USA: Wiley, 2003.

\bibitem{robert2004}
C.~P. Robert and G.~Casella, \emph{Monte Carlo Statistical Methods},
  2nd~ed.\hskip 1em plus 0.5em minus 0.4em\relax Springer, 2004.

\bibitem{devroye1986}
L.~Devroye, \emph{Non-Uniform Random Variate Generation}.\hskip 1em plus 0.5em
  minus 0.4em\relax New York, USA: Springer-Verlag, 1986.

\end{thebibliography}

\clearpage
\appendices
\crefname{section}{Appendix}{Appendices}
\Crefname{section}{Appendix}{Appendices}

\section{Order-Statistics Setup}\label{app:os}
This appendix establishes the order-statistics results that support \cref{subsec:LER_OS}.
\cref{app:os} specifies the statistical model of the sorted reliabilities, \cref{app:cond} turns the two conditional events appearing in \eqref{eq:os_ler_chase}, \eqref{eq:os_ler_res_chase} and \eqref{eq:os_ler_lw} into integrals over these order statistics, and \cref{app:mc} describes how these integrals can be evaluated numerically.

Throughout the appendix, we use $|y_i|$ as the reliability of position $i$.\footnote{The reliability defined in \cref{subsec:channel} is $\alpha_i = |\ell_i| = \frac{2}{\sigma^2}|y_i|$. Since only order relations between reliabilities enter the following expressions and the two differ by the positive constant $\frac{2}{\sigma^2}$, the resulting probabilities are identical.}
The density functions $\fcor(x)$ and $\ferr(x)$ of the reliability at a position $i$ for a correct hard decision or a wrong hard decision, respectively, are given by
\begin{equation*}
\begin{aligned}
\fcor(x) &= \begin{cases}
  \frac{\qsigma(x-1)}{1-\Qsigma(1)} & x\geq 0 \\
  0 & \text{otherwise}
\end{cases} \\
\ferr(x) &= \begin{cases}
  \frac{\qsigma(x+1)}{\Qsigma(1)} & x\geq 0 \\
  0 & \text{otherwise}
\end{cases}
\end{aligned}
\end{equation*}
with
\begin{equation*}
\qsigma(x) = \frac{1}{\sqrt{2\pi\sigma^2}} \, e^{-\frac{x^2}{2\sigma^2}} \text{ and }\Qsigma(x) = \int_x^\infty \qsigma(t) \, dt.
\end{equation*}
The corresponding cumulative distribution functions are denoted by $\Fcor(x)$ and $\Ferr(x)$, respectively, and are given by
\begin{equation*}
\begin{aligned}
\Fcor(x) &= \begin{cases}
  \frac{1-\Qsigma(1)-\Qsigma(x-1)}{1-\Qsigma(1)} & x\geq 0 \\
  0 & \text{otherwise}
\end{cases} \\
\Ferr(x) &= \begin{cases}
  \frac{\Qsigma(1)-\Qsigma(x+1)}{\Qsigma(1)} & x\geq 0 \\
  0 & \text{otherwise}.
\end{cases}
\end{aligned}
\end{equation*}

We now condition on the event $\calE_b$, i.e., the hard decision $\vchd$ is erroneous at exactly $b$ positions.
Since the noise is independent across positions and, by symmetry, every set of $b$ erroneous positions is equally likely, the $b$ reliabilities at erroneous positions are i.i.d.\ with density $\ferr$, the $n-b$ reliabilities at correct positions are i.i.d.\ with density $\fcor$, and the two groups are independent of each other.
This independence is what allows the joint densities below to be multiplied.
The $b$ reliabilities of positions where the hard decision is wrong are denoted by ${\betaos{1}\leq \betaos{2} \leq \ldots \leq \betaos{b}}$.
Similarly, the $n-b$ reliabilities of positions where the hard decision is correct are  denoted by ${\gammaos{1}\leq \gammaos{2} \leq \ldots \leq \gammaos{n-b}}$.
Since the received vector is assumed to be sorted such that ${\alpha_1\leq\ldots\leq\alpha_n}$, position $i$ of the sorted vector is the {$i$-th} smallest element of the merge of the two sequences $\vbeta$ and $\vgamma$.
Every statement about the error vector $\ve$ at a set of sorted positions can therefore be translated into a set of inequalities between the $\betaos{\cdot}$ and $\gammaos{\cdot}$, which is what \cref{app:cond} does.
Here, $\betaos{i}$ is the $i$-th order statistic of the $b$ i.i.d.\ reliabilities with density $\ferr$ and $\gammaos{j}$ is the $j$-th order statistic of the $n-b$ i.i.d.\ reliabilities with density $\fcor$.
Following \cite{david2003}, their marginal densities are given by
\begin{equation*}
\begin{aligned}
  f_{\betaos{i}}(u) &= \Theta(b,i)\, \ferr(u)[\Ferr(u)]^{i-1}[1-\Ferr(u)]^{b-i} \\
  f_{\gammaos{j}}(v) &=  \Theta(n-b,j)\, \fcor(v)[\Fcor(v)]^{j-1}[1-\Fcor(v)]^{n-b-j},
\end{aligned}
\end{equation*}
with $\Theta(g, h)\triangleq g\binom{g-1}{h-1}=\frac{g!}{(h-1)!\,(g-h)!}$ for the ${h\text{-th}}$ order statistic out of $g$ samples.
Consider the subsets ${\calI=\{i_1, \ldots, i_\mu\}\subseteq [b]}$ with ${|\calI|=\mu}$ and ${\calJ=\{j_1,\ldots, j_\nu\}\subseteq [n-b]}$ with $|\calJ|=\nu$, both sorted increasingly.
We denote the joint probability density functions of all $\betabos{b}{i}$ with $i\in\calI$ and of all $\gammabos{b}{j}$ with  $j\in\calJ$ in short by ${f_{\vbeta_{b,\calI}}(\vu_{\calI})\triangleq f_{\{\betaos{i}\}_{i\in\calI}}(\{u_i\}_{i\in\calI})}$ and ${f_{\vgamma_{n-b,\calJ}}(\vv_{\calJ})\triangleq f_{\{\gammaos{j}\}_{j\in\calJ}}(\{v_j\}_{j\in\calJ})}$. Here and in the following, $u_i$ denotes the integration variable associated with $\betaos{i}$ and $v_j$ the one associated with $\gammaos{j}$, both taking non-negative reliability values.
They can be calculated as in \cite{david2003} by
\begin{equation}\label{eq:joint_os_1}
\begin{aligned}
&f_{\vbeta_{b,\calI}}(\vu_{\calI}) = \\
&b!\, \left[\prod_{i\in\calI} \ferr(u_i)\right] \prod_{\ell=0}^{\mu} \left(\frac{[\Ferr(u_{i_{\ell+1}})-\Ferr(u_{i_\ell})]^{i_{\ell+1}-i_{\ell}-1} }{(i_{\ell+1}-i_{\ell}-1)!}\right)
\end{aligned}
\end{equation}
with the conventions $i_0=0$, $i_{\mu+1}=b+1$, $u_{i_0}=0$ and $u_{i_{\mu+1}}=\infty$, as well as
\begin{equation}\label{eq:joint_os_2}
\begin{aligned}
&f_{\vgamma_{n-b,\calJ}}(\vv_{\calJ}) =\\
& (n-b)!\, \left[\prod_{j\in\calJ} \fcor(v_j)\right] \prod_{\ell=0}^{\nu} \left(\frac{[\Fcor(v_{j_{\ell+1}})-\Fcor(v_{j_\ell})]^{j_{\ell+1}-j_{\ell}-1} }{(j_{\ell+1}-j_{\ell}-1)!}\right)
\end{aligned}
\end{equation}
with the conventions $j_0=0$, $j_{\nu+1}=n-b+1$, $v_{j_0}=0$ and $v_{j_{\nu+1}}=\infty$.

We now have all order statistics definitions to show how the terms from \cref{subsec:LER_OS} can be calculated.

\section{Evaluating the Conditional Probabilities}\label{app:cond}

In \eqref{eq:os_ler_chase}, \eqref{eq:os_ler_res_chase}, and \eqref{eq:os_ler_lw}, two types of probabilities need to be calculated.
The first is the probability of $i$ transmission errors $P(\mathcal{E}_i)$ which is determined by \eqref{eq:prob_i_flips}.
The second are conditional probabilities imposing a constraint on the error vector $\ve$ given the event $\mathcal{E}_i$ that $i$ transmission errors have occurred. 
We can calculate them using the above stated order-statistics functions which assume that the event $\mathcal{E}_b$ has occurred.

The conditional probability $P(\wham(\ve_{[p]}) \mathrel{\geq} b-t \mid \calE_{b})$ is required to get the \ac{LER} of Chase-II and restricted Chase decoding patterns  as in \eqref{eq:os_ler_chase}, \eqref{eq:os_ler_res_chase}.
The probability captures the event in which at most $t$ of the $b$ errors fall outside the $p$ least reliable positions, or equivalently, in which at least $b-t$ of them fall within the $p$ least reliable positions.
In other words, the reliabilities $\betaos{1},\ldots,\betaos{b-t}$ have a smaller absolute value than the reliabilities in the positions $[p+1,n]$.
The reliability with smallest absolute value in these positions without an error is $\gammaos{p+1-(b-t)}$, because $b-t$ of the $p$ least reliable positions are erroneous and the remaining $p-(b-t)$ are correct.
Therefore, the event is captured by $\betaos{b-t}\leq \gammaos{p+1-(b-t)}$ and it holds that
\begin{equation*}
\begin{aligned}
      &P(\wham(\ve_{[p]}) \mathrel{\geq} b-t \mid \calE_{b}) \\&= \int_{0}^{\infty}  f_{\betaos{b-t}} (u)\left(\int_{u}^{\infty} f_{\gammaos{p+t+1-b}}(v) \, dv\right) du,
\end{aligned}
\end{equation*}
as also given in \cite{agrawal2000,fossorier2002}.

Before explaining how to calculate the conditional probability $P(\ve_{[\imax_{\pat}]}=\pat_{[\imax_{\pat}]}\mid \calE_{t+w})$, it is helpful to provide an example.

\begin{example}
    Consider a toy example where the $(7,4)$ Hamming code with $t=1$  is decoded with $\npat=19$ logistic weight test patterns\footnote{The avid reader may notice that using $19$ test patterns to decode a perfect code with $16$ codewords does not make sense. However, this is just a toy example for illustration purposes.}.
    The pattern set $\Patlw$ then consists of all patterns with $\wlw(\pat)\leq 7$, listed in \cref{tab:lw_example} in the order $\prec$ of \cref{def:lw} and represented by their set $\calK$ of non-zero positions.
\begin{table}[t]
    \caption{The $\npat=19$ logistic-weight patterns $\pat=\one_{\calK}$ of the toy example, given by their set $\calK$ of non-zero positions and listed in increasing order with respect to $\prec$.}
    \label{tab:lw_example}
    \centering
    \small
    \begin{tabular}{@{}ll@{}}
        \toprule
        $\wlw(\pat)$ & sets $\calK$ of non-zero positions, in order of $\prec$ \\
        \midrule
        $0$ & $\,\,\varnothing$ \\
        $1$ & $\{1\}$ \\
        $2$ & $\{2\}$ \\
        $3$ & $\{3\}$, $\{1,2\}$ \\
        $4$ & $\{4\}$, $\{1,3\}$ \\
        $5$ & $\{5\}$, $\{2,3\}$, $\{1,4\}$ \\
        $6$ & $\{6\}$, $\{2,4\}$, $\{1,5\}$, $\{1,2,3\}$ \\
        $7$ & $\{7\}$, $\{3,4\}$, $\{2,5\}$, $\{1,6\}$, $\{1,2,4\}$ \\
        \bottomrule
    \end{tabular}
\end{table}
    We consider the additional probability of the last pattern, i.e., $\Padd_{\Patlw^\prime}(\pat)$ with $\pat= \one_{\{1,2,4\}}=\begin{bmatrix}
                1,1,0,1,0,0,0
            \end{bmatrix}$ and ${\Patlw^\prime=\Patlw\setminus\{\pat\}}$.
    From \cref{prop:ler_lw_patterns} and \eqref{eq:os_ler_lw}, we know that the additional probability is ${P(\ve_{[\imax_{\pat}]}=\pat_{[\imax_{\pat}]}\mid \calE_{t+w}) P(\calE_{t+w})}$ with $t=1$, $\imax_{\pat}=4$, $w=\wham(\pat)=3$, thus
\begin{equation*}
        \Padd_{\Patlw^\prime}(\pat)=P(\ve_{[4]}=\pat_{[4]}\mid \calE_{4}) P(\calE_{4}).
    \end{equation*}
The probability $P(\calE_{4})$ can be calculated with \eqref{eq:prob_i_flips}.
    To calculate the conditional probability $P(\ve_{[4]}=\pat_{[4]}\mid \calE_{4})$, we consider the order-statistics setup, where $b=4$ transmission errors have occurred.
    The condition $\ve_{[4]}=\pat_{[4]}$ can then be translated to conditions on the corresponding ordered reliabilities $\vbeta$, $\vgamma$ which were flipped and which were not flipped, respectively.
Since $b=4$ and $n-b=3$, the four reliabilities at erroneous positions are $\betabos{4}{1}\leq\ldots\leq\betabos{4}{4}$ and the three at correct positions are $\gammanbos{3}{1}\leq\gammanbos{3}{2}\leq\gammanbos{3}{3}$.
    For the pattern $\pat$, the conditions are illustrated in \cref{fig:os_merge}.
\begin{figure}[htb!]
    \centering
    \begin{tikzpicture}[
    x=1.05cm, y=1cm,
    varnode/.style={font=\footnotesize, anchor=south, inner sep=1.5pt},
    bitnode/.style={font=\footnotesize, anchor=north, inner sep=2pt},
    cond/.style={-{Straight Barb[scale=0.7]}, mittelblau, thick},
]

\fill[hellgrau] (0.55,-0.42) rectangle (4.45,0.52);

\draw[-{Straight Barb[scale=0.8]}, anthrazit] (0.2,0) -- (7.9,0)
    node[anchor=west, font=\footnotesize, anthrazit] {$\alpha$};

\foreach \p/\bit/\var in {%
    1/1/{\betabos{5}{1}},
    2/1/{\betabos{5}{2}},
    3/0/{\gammanbos{2}{1}},
    4/1/{\betabos{5}{3}},
    5/0/{},
    6/0/{},
    7/0/{}}{
    \draw[anthrazit] (\p,-0.1) -- (\p,0.1);
    \node[bitnode] at (\p,-0.14) {$\bit$};
    \node[varnode, mittelblau] at (\p,0.12) {$\var$};
}

\draw[decorate, decoration={brace, amplitude=4pt}, anthrazit]
    (4.75,0.30) -- (7.25,0.30);
\node[font=\footnotesize, anchor=south, anthrazit, align=center] at (6.1,0.44)
    {in any order\\ $\gammanbos{2}{2},\ \betabos{5}{4},\ \betabos{5}{5}$ };

\draw[cond] (4,1.4) -- (6.1,1.4);
\draw[mittelblau, thick] (4,0.94) -- (4,1.4);
\node[font=\footnotesize, anchor=south, mittelblau] at (5.05,1.44)
    {$\betabos{5}{3}\leq\gammanbos{2}{2}$};

\node[font=\footnotesize, anchor=north] at (2.5,-0.46)
    {constrained by $\ve_{[4]}=\pat_{[4]}$};
\node[font=\footnotesize, anchor=north] at (6,-0.46)
    {unconstrained};

\end{tikzpicture}
    \caption{Assignment of the order statistics $\vbeta$ and $\vgamma$ to the sorted positions for the toy example with $\pat=\one_{\{1,2,4\}}$, $b=4$ and $n-b=3$. The bottom row is $\pat$. Walking through the positions, a one is associated with the next $\beta$ and a zero with the next $\gamma$. Only the shaded positions $[\imax_{\pat}]=[4]$ are constrained by $\ve_{[4]}=\pat_{[4]}$, plus the closing condition linking position $4$ to the smallest correct reliability that is not constraint.}
    \label{fig:os_merge}
\end{figure}
The desired constraint $\ve_{[4]}=\pat_{[4]}$ can only constrain reliabilities in the positions $1$ to $4$.
The order of the last three reliabilities $\gammanbos{3}{2}$, $\gammanbos{3}{3}$ and $\betabos{4}{4}$ does not play a role since the condition $\ve_{[4]}=\pat_{[4]}$ does not cover these positions.
One further condition is needed at the largest non-zero position of $\pat$, i.e., position $4$, which has to be less reliable than each of the positions $5$ to $7$.
Only the correct positions $5$ to $7$ give a non-trivial condition, and since positions $1$ to $4$ already contain one of the three correct positions, the smallest reliability of a correct position beyond position $4$ is $\gammanbos{3}{2}$.
Hence, $\betabos{4}{3}\leq\gammanbos{3}{2}$ is required.
Some of the conditions are fulfilled by definition of order statistics, e.g., $\betabos{4}{1}\leq \betabos{4}{2}$.
Overall, the three conditions $\betabos{4}{2}\leq \gammanbos{3}{1}$, $\gammanbos{3}{1}\leq \betabos{4}{3}$ and $\betabos{4}{3}\leq \gammanbos{3}{2}$ have to hold.
We denote them by a set of conditions $\Scond$, i.e.,
\begin{equation*}
    \Scond = \{\betabos{4}{2}\leq \gammanbos{3}{1}, \gammanbos{3}{1}\leq \betabos{4}{3}, \betabos{4}{3}\leq \gammanbos{3}{2}\}.
\end{equation*}
These conditions involve the two $\beta$ variables with indices $\Icond=\{2,3\}$ and the two $\gamma$ variables with indices $\Jcond=\{1,2\}$, thus $\df=|\Icond|+|\Jcond|=4$ variables have to be considered in the following integral.
The set of non-negative tuples for which the conditions in $\Scond$ hold is denoted by $\Rcond{\df}$. We then have
\begin{equation*}
\begin{aligned}
        &P(\ve_{[4]}=\pat_{[4]}\mid \calE_{4}) =\\ &\int\limits_{\Rcond{4}} f_{\vbeta_{4,\{2,3\}}}(\vu_{\{2,3\}})f_{\vgamma_{3,\{1,2\}}}(\vv_{\{1,2\}}) \, d\vu_{\{2,3\}}d\vv_{\{1,2\}},
\end{aligned}
\end{equation*}
where $f_{\vbeta_{4,\{2,3\}}}(\vu_{\{2,3\}})$ is the joint probability density function of $\betabos{4}{2} $ and $ \betabos{4}{3}$, and $f_{\vgamma_{3,\{1,2\}}}(\vv_{\{1,2\}})$ is the joint probability density function of $\gammanbos{3}{1}$ and $\gammanbos{3}{2}$, both given by \eqref{eq:joint_os_1} and \eqref{eq:joint_os_2}.

\end{example}

Keeping this example in mind, we can now generalize the calculation of $P(\ve_{[\imax_{\pat}]}=\pat_{[\imax_{\pat}]}\mid \calE_{t+w})$ for a pattern $\pat\in\Patlw\setminus\{\zero\}$ with Hamming weight $w>0$.

\begin{calc}\label{calc:lw}
Consider a pattern $\pat\in\Patlw\setminus\{\zero\}$  with Hamming weight $w>0$ as in \cref{prop:ler_lw_patterns} within the order-statistics setup with $b=t+w$ transmission errors, and write $\rho_i$ for the $i$-th entry of $\pat$.
The set $\Scond$ of conditions imposed by $\ve_{[\imax_{\pat}]}=\pat_{[\imax_{\pat}]}$ is obtained as follows.
\begin{enumerate}
    \item Assign an order-statistics variable to every position $i\in[\imax_{\pat}]$: if $\rho_i=1$, position $i$ carries $\betaos{r}$ with $r=\wham(\pat_{[i]})$; if $\rho_i=0$, position $i$ carries $\gammaos{s}$ with $s=i-\wham(\pat_{[i]})$.
    \item Chain the assigned variables by $\leq$ in order of increasing position. Since the received vector is sorted by reliability, this chain is exactly the statement $\ve_{[\imax_{\pat}]}=\pat_{[\imax_{\pat}]}$.
    \item Append the closing condition $\betaos{w}\leq\gammabos{b}{\imax_{\pat}-w+1}$.
    Position $\imax_{\pat}$ carries $\betaos{w}$ because $\rho_{\imax_{\pat}}=1$, and $\imax_{\pat}-w$ of the positions in $[\imax_{\pat}]$ are not flipped, so the smallest reliability of a correct position outside $[\imax_{\pat}]$ is $\gammabos{b}{\imax_{\pat}-w+1}$.
    \item Drop all conditions between two variables of the same type, since $\betaos{r}\leq\betaos{r+1}$ and $\gammaos{s}\leq\gammaos{s+1}$ hold by definition of the order statistics.
\end{enumerate}
The indices of the order-statistics variables $\beta$ and $\gamma$ that are constrained by $\Scond$ are denoted by $\Icond$ and $\Jcond$, respectively. We define $\df\triangleq |\Icond|+|\Jcond|$, which is the total number of constrained variables, and satisfies $\df\leq\imax_{\pat}$.
The conditional probability $P(\ve_{[\imax_{\pat}]}=\pat_{[\imax_{\pat}]}\mid \calE_{t+w})$ can then be calculated by
\begin{equation}\label{eq:lw_os_calc}
    \int\limits_{\Rcond{\df}} f_{\vbeta_{b,\Icond}}(\vu_{\Icond})f_{\vgamma_{n-b,\Jcond}}(\vv_{\Jcond}) \, d\vu_{\Icond}d\vv_{\Jcond},
\end{equation}
where $\Rcond{\df}$ denotes the subset of $\mathbb{R}^{\df}_{\geq0}$ for which the conditions $\Scond$ are fulfilled. The joint probability density functions  $f_{\vbeta_{b,\Icond}}(\vu_{\Icond})$ and $f_{\vgamma_{n-b,\Jcond}}(\vv_{\Jcond})$ are the ones defined in \eqref{eq:joint_os_1} and \eqref{eq:joint_os_2}, and they may be multiplied because the two groups of reliabilities are independent given $\calE_b$.
\end{calc}

\section{Monte Carlo Evaluation of the Order-Statistics Integrals}\label{app:mc}

Solving multidimensional integrals, such as the ones in \eqref{eq:lw_os_calc}, numerically based on the multidimensional functions is complex. 
Therefore, we calculate the integral more efficiently using Monte Carlo integration \cite{robert2004}.
Assume that we can sample from the joint distributions $f_{\vbeta_{b,\Icond}}(\vu_{\Icond})$ and $f_{\vgamma_{n-b,\Jcond}}(\vv_{\Jcond})$ of \eqref{eq:joint_os_1} and \eqref{eq:joint_os_2}.
We can then estimate the probability via Monte Carlo simulation where in each step we draw one sample from each of the two distributions, independently as justified before, and check whether the conditions $\Scond$ are fulfilled.
The probability is then the number of tries where $\Scond$ was fulfilled over the number of overall tries.
Note that only the conditional probabilities are estimated in this way, whereas the probability $P(\calE_b)$ of the conditioning event is evaluated in closed form with \eqref{eq:prob_i_flips}.
The estimator therefore does not suffer from the rare-event problem of a direct simulation as the conditional probabilities are not small, even where the resulting \ac{LER} is of the order of $10^{-6}$.

Sampling from the joint distributions can be done using the following property from \cite{devroye1986}. The order statistics ${x_{(1)}\leq\ldots\leq x_{(m)}}$ that correspond to the sequence $x_1,\ldots,x_{m}$ of i.i.d. random variables with \ac{CDF} $F$ can be generated as
\begin{equation*}
    F^{-1}(U_{(1)}), \ldots, F^{-1}(U_{(m)}),
\end{equation*}
where $F^{-1}$ is the inverse \ac{CDF} and $U_{(1)}\leq\ldots\leq U_{(m)}$ are the order statistics of $m$ i.i.d.\ $\mathrm{Unif}(0,1)$ random variables.
The \acp{CDF} $\Ferr$ and $\Fcor$ can be inverted in closed form as
\begin{equation*}
\begin{aligned}
    \Ferr^{-1}(U) &= \Qsigma^{-1}\!\left(\Qsigma(1)(1-U)\right) - 1, \\
    \Fcor^{-1}(U) &= \Qsigma^{-1}\!\left(\left(1-\Qsigma(1)\right)(1-U)\right) + 1,
\end{aligned}
\end{equation*}
for an argument $0\leq U\leq 1$, so that a sample only requires one evaluation of the inverse Gaussian tail function.
The uniform order statistics are drawn sequentially using the fact that $U_{(k)}$ has the marginal distribution $\mathrm{Beta}(k, m+1-k)$ and that, given $U_{(k-1)}$, the remaining order statistics are those of $m-k+1$ i.i.d. $\mathrm{Unif}(U_{(k-1)},1)$ variables.
Only the order statistics with indices in $\Icond$ and $\Jcond$ have to be generated, since \eqref{eq:lw_os_calc} does not involve the others.

\end{document}